\documentclass[conference]{IEEEtran}

\usepackage{amsmath, amssymb}
\usepackage{url}
\usepackage{cite}

\usepackage{amsthm, graphicx, color}
\usepackage[font=small]{caption}
\usepackage[font=footnotesize]{subcaption}

\newtheorem{theorem}{Theorem}
\newtheorem{lemma}{Lemma}
\newtheorem{proposition}{Proposition}

\definecolor{erdem}{rgb}{0,0,0}
\title{Performance Gains of Optimal Antenna Deployment for Massive MIMO Systems\vspace{-10pt}}
\author{    \IEEEauthorblockN {Erdem Koyuncu    }
    \IEEEauthorblockA {
        Department of Electrical and Computer Engineering, University of Illinois at Chicago \vspace{-20pt}
    }}

\begin{document}
\maketitle
\begin{abstract}
We consider the uplink of a single-cell multi-user multiple-input multiple-output (MIMO) system with several single antenna transmitters/users and one base station with $N$ antennas in the $N\rightarrow\infty$ regime. The base station antennas are evenly distributed to $n$ admissable locations throughout the cell. 

First, we show that a reliable (per-user) rate of $O(\log n)$ is achievable through optimal locational optimization of base station antennas. We also prove that an $O(\log n)$ rate is the best possible. Therefore, in contrast to a centralized or circular deployment, where the achievable rate is at most a constant, the rate with a general deployment can grow logarithmically with $n$, resulting in a certain form of ``macromultiplexing gain.'' 

Second, using tools from high-resolution quantization theory, we derive an accurate formula for the best achievable rate given any $n$ and any user density function. According to our formula, the dependence of the optimal rate on the user density function $f$ is curiously only through the differential entropy of $f$. In fact, the optimal rate decreases linearly with the differential entropy, and the worst-case scenario is a uniform user density. Numerical simulations confirm our analytical findings. 
\end{abstract}
\vspace{-10pt}
\section{Introduction}
\label{secIntro}
Massive multiple-input multiple-output (MIMO) is a recently-proposed \cite{marzetta1} promising method that can greatly improve the energy efficiency, spectral efficiency, and coverage of cellular networks \cite{swindletutorial}. The defining feature of massive MIMO is the large number of base station antennas, which naturally provides spatially-orthogonal channels to different users. This enables interference-free inter-cell communication.

The existence of many antennas comes with the opportunity of distributing them to different geographical locations. The idea of utilizing such a distributed antenna system (DAS) as opposed to a colocated antenna system (CAS) predates massive MIMO or even ``regular'' MIMO itself, and goes back to at least \cite{saleh_dist_antennas}. Another early work \cite{kjkerpez} studies the signal-to-noise ratios (SNRs) and the signal-to-interference ratios (SIRs) for a DAS where each antenna transmits the same message. Due to reduced distance between the users and the base station, a DAS can significantly improve coverage and power-efficiency, and even reduce intra-cell interference \cite{wanchoi1}.

With the advent of multi-user MIMO and massive MIMO techniques, DAS design and analysis have drawn renewed interest. Achievable rates for a circular DAS and the corresponding optimal antenna radii have been studied in \cite{j_gan_circular_das2, j_gan_circular_das,yindi_jing_circular_das, kamga1}. These works show that using circular DAS can greatly improve upon a CAS in terms of the sum rate. However, the improvement is only a fixed constant and does not grow asymptotically with the number of distinct antenna locations. Lloyd-algorithm based numerical approaches to antenna location optimization can be found in \cite{wang_ant_loc_opt,wang_ant_loc_opt2}. Another numerical approach to optimal antenna placement based on stochastic approximation theory can be found in \cite{firouz1}. The performance of circular antenna arrays with or without a central antenna has been studied in \cite{circular_das_with_central_antenna}. In \cite{lin_dai}, upper bounds on the single-user capacity and approximations to the sum capacity have been obtained for both randomly-distributed and colocated antenna systems. It is demonstrated that distributed antennas can provide significant gains over colocated antennas, especially when the users have channel state information (CSI). In \cite{liu_antenna}, antenna locations are optimized to maximize a lower bound on the energy efficiency of a cellular DAS with two types of users. In \cite{zhang}, the authors provide deterministic approximations for the sum-rate of massive MIMO uplink with distributed and possibly-correlated antennas.

Most of the above work have focused on a fixed antenna topology, such as a circular DAS. Clearly, a better performance can be achieved by considering base station antennas in general position. However, in a massive MIMO system, allowing all the base station antennas to be in different locations may not be feasible due to costs of  backhaul implementation and land. One novelty of this work is to consider a model where $N$ base station antennas are evenly distributed to a fixed number $n$ of distinct locations in the $N\rightarrow\infty$ regime. This allows us to obtain the benefits of a distributed massive MIMO system, while avoiding potentially-high implementation costs. 

Also, most existing work on antenna location optimization relies on numerical methods. To the best of our knowledge, when the base station antennas of a massive MIMO system are allowed arbitrary positions, a formal analysis of the benefits of antenna location optimization is not available in the literature. In this work, we present such an analysis with a particular focus on determining the achievable per-user rates. 

The rest of this paper is organized as follows: In Section \ref{secSystemModel}, we introduce the system model. In Section \ref{secAchievability}, we present rigorous bounds on the achievable rates with optimized antenna locations. In Section \ref{secHighRes}, we provide a quantization-theoretical analysis of the best possible rates. In Section \ref{secNumerical}, we provide simulation results. Finally, in Section \ref{secConclusions}, we draw our main conclusions. Due to space limitations, some of the technical proofs are provided in the extended version \cite{massivemimotechrep} of this paper.

{\it Notation:} $\|\cdot\|$ is the Frobenius norm. $A^k$ is the $k$th Cartesian power and $|A|$ is the cardinality of a set $A$. Given $b\in\mathbb{R}^d$ and $A \subset \mathbb{R}^d$, we let $A + b \triangleq \{x + b:x\in A\}$. $A-b$ and $bA$ are similar. Given sets $A$ and $B$, $A\backslash B$ is the set of elements of $A$ that are not in $B$. $\mathcal{CN}(0,\alpha)$ is a circularly-symmetric complex Gaussian random variable with variance $\frac{\alpha}{2}$ per real dimension.  $\log(\cdot)$ is the natural logarithm, $\Gamma(\cdot)$ is the Gamma function. For a logical statement $T$, we let $\mathbf{1}(T) = 1$ is $T$ is true, and $\mathbf{1}(T) = 0$, otherwise. $\mu(A)$ is the Lebesgue measure of a set $A$. For $x\in\mathbb{R}$, we let $x^+ \triangleq x \mathbf{1}(x \geq 0)$. $O(\cdot)$ and $o(\cdot)$ are the standard Bachmann-Landau symbols. For sequences $a_n,\,b_n$, the notation $a_n \sim b_n$ means $\lim_{n\rightarrow\infty} \frac{a_n}{b_n} = 1$.
\section{System Model}
\label{secSystemModel}
We consider a single-cell multi-user distributed MIMO system with $N$ base station antennas. The base station antennas are divided to $n$ groups, where each group consists of $\frac{N}{n}$ antennas. We assume that $N$ is a multiple of $n$. This ensures that the number of antennas at each group is an integer. 

Given $i\in\{1,\ldots,n\}$, suppose that the antennas in Group $i$ are all located at the same point $x_i \in \mathbb{R}^d$, where the ambient dimension $d$ is a positive integer (In practice, one is usually interested in the case $d\in\{1,2,3\}$.). Therefore, the $N$ base station antennas are geographically distributed to $n$ co-located antenna groups, where each group consists of $\frac{N}{n}$ antennas. The practical motivation for such a colocated-distributed architecture has been described in Section \ref{secIntro}.

We study an uplink scenario where $m$ single-antenna users simultaneously wish to communicate with the base station. Specifically, given $j\in\{1,\ldots,m\}$, User $j$ wishes to communicate the complex Gaussian symbol $s_j \sim \mathcal{CN}(0,1)$ to the base station. Following the massive MIMO literature, we assume that each user transmits with normalized power $P/N$. Thus, User $j$ transmits the signal $s_j \sqrt{P/N}$ over its antenna. 

Given $i\!\in\!\{1,\ldots,n\}$, $j\in\{1,\ldots,m\}$, and $\ell\in\{1,\ldots,\frac{N}{n}\}$, let $h_{ji\ell}\in\mathbb{C}$ be the channel gain between User $j$ and Antenna $\ell$ of Group $i$. We assume that
$h_{jil} \sim \mathcal{CN}(0,\|x_i - u_j\|^{-r})$, where $r \geq 0$ is the path-loss exponent, and  $u_j\in\mathbb{R}^d$ is the location of User $j$.  The channel input-output relationships are $y_{i\ell} = \sum_{j=1}^m h_{ji\ell} s_j \sqrt{P/N} + \tau_{i \ell}$, where $\tau_{i\ell} \sim \mathcal{CN}(0,1)$ is the noise at the $\ell$th antenna of Group $i$. We assume that all the channel gains, noises, and the data symbols are independent. 

We consider the massive MIMO regime where the number of base station antennas $N$ grows to infinity. Meanwhile, the number of users $m$ and the number of antenna groups $n$ remain fixed $N$-independent constants. Under such conditions, using zero forcing at the base station, a reliable data rate of 
\begin{align}
\label{aksjdlak020483}
\log\biggl( 1+ P \lim_{N\rightarrow\infty} \frac{1}{N} \sum_{i = 1}^{n} \sum_{\ell = 1}^{\frac{N}{n}} \|x_i - u_j\|^{-r} \biggr) \frac{\mathrm{nats}}{\mathrm{sec}\cdot\mathrm{Hz}}
\end{align}
is achievable for User $j$, as $N\rightarrow\infty$ \cite[Proposition 1]{yindi_jing_circular_das}. Noting that the summand in (\ref{aksjdlak020483}) is independent of $\ell$, the achievable rate in (\ref{aksjdlak020483}) for a generic user at location $u$ simplifies to
\begin{align}
\label{aksjdlak0204832}
R(u,\mathcal{X}) \triangleq \log\left( 1+  \frac{P}{n} \sum_{i = 1}^{n} \frac{1}{\|x_i - u\|^r} \right),
\end{align}
where $ \mathcal{X} \triangleq \{x_1,\ldots,x_n\}$ is the set of the distinct locations of base station antennas. The achievabilities of (\ref{aksjdlak020483}) or (\ref{aksjdlak0204832}) require the typical assumptions of full CSI at the base station and perfect cooperation among the base station antennas. 

We model the user location as a random variable with a certain probability density function $f$ over the cell. We also assume that the cell shape $S$ is bounded. Thus, $S \subset [0,M]^d$, where $M \geq 0$ is finite, and the density $f$ vanishes outside $S$. Given a deployment $\mathcal{X}$ of base station antennas, the average rate of a typical user can then be expressed as
\begin{align}
\label{arofeks}
R(\mathcal{X}) & \triangleq \int_{S} R(u,\mathcal{X}) f(u) \mathrm{d}u \\
\label{actualrakakakeks} & = \int_S \log\left( 1+  \frac{P}{n} \sum_{i = 1}^{n} \frac{1}{\|u - x_i\|^r} \right) f(u)\mathrm{d}u
\end{align}

Our goal is to find the optimal deployment of antennas that maximize the average rate (In order to avoid unnecessary technicalities, we assume that such an optimal deployment exists.). In other words, we wish to determine
\begin{align}
\label{oieroqiuewiquweq}
R_n \triangleq \max_{\mathcal{X}} R(\mathcal{X}),
\end{align}
and the structure of deployments that can achieve (\ref{oieroqiuewiquweq}). 

\section{Bounds on the Best Per-User Rates}
\label{secAchievability}
We first present the achievable performance gains through an optimal deployment of antennas. Clearly, the best per-user rate $R_n$ is a monotonically non-decreasing function of $n$. This leads to the lower bound $R_n \geq R_1$, which is achievable when all the base station antennas are colocated to a certain unique optimal location. A natural starting point may be to improve upon this best-possible rate $R_1$ with colocated antennas.

It has been previously established \cite{j_gan_circular_das2, j_gan_circular_das,yindi_jing_circular_das, kamga1} that a circular array of antennas can improve upon the rate $R_1$. Unfortunately, the improvement is only an additive constant. We will show in the following that $R_n$ grows at least proportionally to $\log n$. Therefore, the per-user rate can, in fact, be arbitrarily large.

As also suggested several times in different contexts (see, e.g., \cite{lin_dai}), the idea is to distribute the antenna groups ``evenly'' throughout the cell. For example, in two dimensions, we can position the $n$ groups of base station antennas in a uniform square lattice so that, for any user at location $u$, there exists a group of base station antennas that is at most $O(n^{-\frac{1}{2}})$-far. According to (\ref{aksjdlak0204832}), this guarantees a rate of $O(\log(1+Pn^{\frac{r}{2}-1}))$, regardless of the location of the user. It follows that for $r > 2$, the per-user rate is unbounded as $n\rightarrow\infty$, growing at least as $(\frac{r}{2}-1)\log n$. In contrast, for colocated and circular topologies, the nearest base station antenna is $O(1)$-apart for a positive fraction of user locations, which leads to a bounded rate.

Extension of the above idea of a uniform square lattice to any dimension leads to the following result.

\begin{proposition}
\label{achievabilityprop}
For any $n \geq 1$, we have
\begin{align}
R_n \geq  \log\left(1 + \frac{P n^{\frac{r}{d}-1}}{ M^r} \right).
\end{align}
\end{proposition}
\begin{proof}
The proof is provided in \cite[Appendix A]{massivemimotechrep}.
\end{proof}

According to Proposition \ref{achievabilityprop}, we can observe that $R_n$ grows at least as $(\frac{r}{d}-1)^+\log n + O(1)$. Therefore, the optimal per-user rate grows at least logarithmically with the number of admissable antenna locations, indicating an unbounded rate improvement over centralized or circular antenna arrangements. A fundamental question is then whether the $O(\log n)$ growth of the rate is the best possible. The affirmative answer is provided by the following result.

\begin{proposition}
\label{converseprop}
Let $\eta \triangleq \mu(S) \sup_{u\in S} f(u)$. We have
\begin{align}
R_n \leq  \frac{\eta r}{d} \log n + \eta\log\left(1+\frac{P}{(\mu(S))^\frac{r}{d}}\right) + 2\eta r,\,\forall n\geq 1.
\end{align}
\end{proposition}
\begin{proof}
The proof is provided in \cite[Appendix B]{massivemimotechrep}.
\end{proof}
Together, Propositions \ref{achievabilityprop} and \ref{converseprop} reveal that $R_n$ grows exactly logarithmically with the number of antenna locations $n$. This motivates us to capture the constant multiplicative factor of $\log n$ that appears in the expression of $R_n$ through the quantity
\begin{align}
\label{macromultiplexing}
\rho \triangleq \lim_{n\rightarrow\infty} \frac{R_n}{\log n}.
\end{align}
We have assumed that the limit exists for a simpler exposition. 

From a mathematical point of view, the definition in (\ref{macromultiplexing}) is very similar to the definition of the multiplexing gain of a regular MIMO system. The only difference is that (\ref{macromultiplexing}) describes the factor of $\log n$ in the data rate formula, where as the multiplexing gain usually refers to the factor of $\log P$. The quantity $\rho$ can also be considered as a certain form of macrodiversity gain as it emerges due the macroscale (much larger than carrier wavelength) separation of antennas. Since, to the best of our knowledge, there is no unique established definition of a macrodiversity gain, and since the nature of $\rho$ is very similar to that of the multiplexing gain, we refer to $\rho$ as the macromultiplexing gain of the system.

The macromultiplexing gain provides a precise asymptotic characterization of the gains that are achievable through locational optimization of antennas. The following theorem, which immediately follows from Propositions \ref{achievabilityprop} and \ref{converseprop}, describes the bounds on $\rho$ that we have obtained so far.
\begin{theorem}
\label{theoremi}
For any user density $f$, we have
\begin{align}
\label{pquwepuqwpe}
\left(\frac{r}{d}-1\right)^+ \leq \rho \leq \frac{\eta r}{d},
\end{align}
where $\eta = \mu(S)\sup_{u\in S}f(u)$, and $S\subset\mathbb{R}^d$ is the cell.
\end{theorem}
Note that $\eta \geq 1$ with equality if $f$ is uniform on $S$. Thus, for a uniform $f$, the upper and lower bounds on $\rho$ differ by only unity. In general, depending on how ``non-uniform'' $f$ is, the gap between the bounds may be larger than $1$. It does not seem straightforward to close this gap even in the case of a uniform distribution, at least by completely formal means. On the other hand, our analysis in the next section suggest that, in fact, $\rho = \left(\frac{r}{d}-1\right)^+$,  regardless of the user density.
 
\section{Analysis of the Best Per-User Rates using high-resolution Quantization Theory}
\label{secHighRes}
The bounds on $R_n$ in Propositions \ref{achievabilityprop} and \ref{converseprop} do not provide much insight on the exact value of $R_n$ for a given finite $n$, as they are not tight in general. Also, despite the fact that the bounds provide the asymptotic formula $R_n= \rho \log n + O(1)$ via Theorem \ref{theoremi}, they are not strong enough to reveal the exact value of $\rho$, or the nature of $O(1)$ term.

The goal of this section is to find an accurate formula for $R_n$. To achieve accuracy, we utilize certain approximations that hold for large values of the design parameters $n$ and $r$. Each step of approximation is described in the following together with the corresponding justifications and analysis. 

\subsection{Reduction to a Quantization Problem}
\label{secReduction}
The first step involves approximating the cost function in (\ref{actualrakakakeks}) via another function that allows a quantization-theoretical approach. For this purpose, we begin by noting the bound
\begin{align}
\sum_{i = 1}^{n} \frac{1}{\|u - x_i\|^r} \geq \frac{1}{\min_{i\in\{1,\ldots,n\}}\|u - x_i\|^r}
\end{align}
that holds for any $u$ and $\mathcal{X}$. When searching for an optimal set of locations, our idea is to use the approximation
\begin{align}
\label{poqwiepoqiwepoqw}
\sum_{i = 1}^{n} \frac{1}{\|u - x_i\|^r} \simeq \frac{1}{\min_{i\in\{1,\ldots,n\}}\|u - x_i\|^r}.
\end{align}
Clearly, (\ref{poqwiepoqiwepoqw}) is potentially inaccurate for an arbitrary $\mathcal{X}$. For example, if all $x_i$ are equal, then the left hand side of (\ref{poqwiepoqiwepoqw}) is $n$ times the right hand side. However, as also discussed in the beginning of Section \ref{secAchievability}, we expect the optimal antenna locations to not be superimposed, but instead be evenly distributed throughout the cell $S$ through some lattice-like structure. For such scenarios, the following lemma shows that (\ref{poqwiepoqiwepoqw}) is, in fact, an asymptotic equality as $r\rightarrow\infty$.

\begin{lemma}
\label{latticelemma}
Let $\Lambda = \{\beta B x:x\in\mathbb{Z}^d\}$ be a lattice in $\mathbb{R}^d$, where $\beta > 0$, and $B\in\mathbb{R}^{d\times d}$ has unit determinant. For $r>d$, there is a constant $K \geq 1$ that is independent of $\beta$, and a set $A$ of measure zero such that for every $u\in\mathbb{R}^d \backslash A$, we have
\begin{align}
\label{lattlemma1}
\sum_{x\in \Lambda} \frac{1}{\|u - x\|^r} \leq \frac{K}{\min_{x\in \Lambda} \|u - x\|^r}.
\end{align}
Moreover, for every $u\in\mathbb{R}^d \backslash A$, we have
\begin{align}
\label{lattlemma2}
\sum_{x\in \Lambda} \frac{1}{\|u - x\|^r} \sim \frac{1}{\min_{x\in \Lambda} \|u - x\|^r} \mbox{ as } r\rightarrow\infty.
\end{align}
\end{lemma}
\begin{proof}
The proof is provided in \cite[Appendix C]{massivemimotechrep}.
\end{proof}
As a result of (\ref{lattlemma2}), for antenna locations $\mathcal{X}$ of interest, we may approximate $R(\mathcal{X})$ in (\ref{actualrakakakeks}) via
\begin{align}
\label{oiquweoiqwueoiquweoiquweq}
R(\mathcal{X})\! \simeq \! \int_S \! \log\left( 1\!+\!  \frac{P}{n} \frac{1}{\min_{1\leq i \leq n}\|u - x_i\|^r} \right) f(u)\mathrm{d}u,
\end{align}
We expect (\ref{oiquweoiqwueoiquweoiquweq}) to be tight as $r\rightarrow\infty$. For a finite $r$, according to (\ref{lattlemma1}), we expect the approximation error to be at most a constant $n$-independent rate. Numerical results in Section \ref{secNumerical} suggest that the error is negligible even for small $r$.

Second, as also shown in the proof of Proposition \ref{achievabilityprop} in \cite[Appendix A]{massivemimotechrep}, the term $n \min_{1\leq i \leq n} \|u-x_i\|^r$ in (\ref{oiquweoiqwueoiquweoiquweq}) can grow at fast as $O(n^{\frac{r}{d}-1})$, which means that the ``$1+$'' term inside the logarithm will be insignificant. As a result, we have
\begin{align}
\label{oiquweoiqwueoiquweoiquweq2}
R(\mathcal{X}) \simeq  r \widetilde{R}(\mathcal{X}) + \log\frac{P}{n},
\end{align}
where
\begin{align}
\label{qopmcn3}
 \widetilde{R}(\mathcal{X}) \triangleq \int_S \log \frac{1}{\min_{i\in\{1,\ldots,n\}}\|u - x_i\|}  f(u)\mathrm{d}u.
\end{align}
Maximizing $ \widetilde{R}(\mathcal{X})$ over all $\mathcal{X}$ as $\widetilde{R}_n \!\triangleq\! \max_{\mathcal{X}} \widetilde{R}(\mathcal{X})$, we obtain
\begin{align}
\label{poqwjepqoweq}
R_n \simeq r \widetilde{R}_n + \log\frac{P}{n}.
\end{align}
We now make a slight digression into a different but very relevant problem. Given $\mathcal{X} = \{x_1,\ldots,x_n\}$ as usual, let 
\begin{align}
\label{classicalqtheory}
 \overline{R}(\mathcal{X}) \triangleq \int_S \min_{i\in\{1,\ldots,n\}}\|u - x_i\|^r  f(u)\mathrm{d}u.
\end{align}
Determining the performance and structure of the minimizers of $\overline{R}(\mathcal{X})$ is the fundamental problem of quantization theory and has been the subject of many publications \cite{quantization}. Exact solutions are out of reach for a given arbitrary $n$ and density $f$, but tight asymptotic results are known for the so-called high-resolution regime $n\rightarrow\infty$, see. e.g. \cite{zador1}.

The problem of maximizing $\widetilde{R}(\mathcal{X})$ in (\ref{qopmcn3}) is very similar to the problem of minimizing $\overline{R}(\mathcal{X}) $ in (\ref{classicalqtheory}). The only difference is that in (\ref{qopmcn3}), one considers a monotonic function (logarithm) of the minimum of norms instead of a mere minimum of $r$th powers of norms in (\ref{classicalqtheory}). It is thus natural to expect that the methods and ideas for minimizing (\ref{classicalqtheory}) can successfully be applied for maximizing (\ref{qopmcn3}). In the following, we first consider a uniform user density over a cubic cell $S = [0,M]^d,\,M>0$. We shall then proceed to a general non-uniform density $f$.

\subsection{Solution to the Quantization Problem: Uniform Case}
\label{secHighRes1}
We begin with some definitions that will be useful for our discussion in this section. Given $\mathcal{X} = \{x_1,\ldots,x_n\}$, let
\begin{multline}
V_i(\mathcal{X}) \triangleq \bigl\{u \in S:\|u - x_i\| \leq \|u - x_j\|,\, \\ \forall j\in\{1,\ldots,n\}\bigr\},\,i\in\{1,\ldots,n\}
\end{multline}
denote the Voronoi cells that are generated by $\mathcal{X}$. We have
\begin{align}
\label{voronoidecomposition}
\widetilde{R}(\mathcal{X}) = \sum_{i=1}^n \int_{V_i} \log \frac{1}{\|u - x_i\|} f(u) \mathrm{d}u
\end{align}

Also, given $A\subset\mathbb{R}^d$ with centroid $c(A) \!\triangleq\! \frac{1}{\mu(A)}\int_A u \mathrm{d}u$, let
\begin{align}
\label{logmoment}
\zeta(A) = \frac{1}{\mu(A)} \int_{A-c(A)} \log \frac{1}{\|u\|} \mathrm{d}u + \frac{1}{d}\log \mu(A)
\end{align}
denote the normalized logarithmic moment of $A$. The moment $\zeta(A)$ is normalized in the sense that $ \zeta (\alpha A) = \zeta(A),\,\forall\alpha > 0$. In particular, we let $\zeta_{b,d} \triangleq \zeta([0,1]^d)$ and $\zeta_{s,d} \triangleq \zeta(\{x:\|x\| \leq 1 \})$ and  denote the normalized logarithmic moment of the $d$-dimensional ball and the $d$-dimensional cube, respectively. 

By using, for example, generalized spherical coordinates, it can be shown that for any $ d\geq 1$, we have
$\zeta_{b,d} = \frac{1}{d}\log \frac{e \pi^{d/2}}{\Gamma(1+\frac{d}{2})}$. It does not appear as straightforward to obtain a general formula for $\zeta_{s,d}$. Still, we can calculate the special cases $\zeta_{s,1}  = \log (2e)$ and $\zeta_{s,2} = \frac{1}{4}(6 + 2\log 2 - \pi)$. 

One important consequence of the moment definition in (\ref{logmoment}) is provided by the following proposition.
\begin{proposition}
\label{interpreterprop}
Suppose that the Voronoi cells $V_i(\mathcal{X})$ generated by $\mathcal{X}$ are all congruent to a certain polytope $T$, and each Voronoi cell has the same $d$-dimensional volume. Then, 
\begin{align}
\label{interpreter}
\widetilde{R}(\mathcal{X}) = \frac{1}{d}\log \frac{n}{M^d} +   \zeta (T).
\end{align}
\end{proposition}
\begin{proof}
The result follows immediately from (\ref{voronoidecomposition}) and  (\ref{logmoment}). \!\!\!\!\!
\end{proof}
Proposition \ref{interpreterprop} will be very useful for interpreting the conclusions of the following result.

\begin{proposition}
\label{approxunifdim2prop}
Let $f(u) = \frac{1}{M^d}\mathbf{1}(u\in[0,M^d])$. For the case of one dimension $d=1$, for every $n \geq 1$, we have
\begin{align}
\label{generalasymptoticformulad1} \widetilde{R}_n = \log \frac{n}{M} +   \zeta_1^{\star},
\end{align}
where $\zeta_1^{\star} = \zeta_{b,1}  = \zeta_{s,1} = \log(2e) = 1.693147\cdots$ 
is the normalized logarithmic moment of the interval.

 For $d \geq 2$, as $n\rightarrow\infty$, we have
\begin{align}
\label{generalasymptoticformula} \widetilde{R}_n = \frac{1}{d}\log \frac{n}{M^d} +   \zeta_d^{\star} + o(1),
\end{align}
for some constant $\zeta_d^{\star}$ that satisfies
\begin{align}
\label{pwqoepoqwiewoqieqw}
\zeta_{s,d} \leq \zeta_d^{\star} \leq \zeta_{b,d}
\end{align}
and depends only on the dimension. In particular,
\begin{align}
\label{keowkeowke}
\zeta_2^{\star} = \frac{3}{2} - \frac{\pi}{2\sqrt{3}} + \frac{1}{4} \log\frac{27}{4} = 1.070485\cdots.
\end{align}
 is the normalized logarithmic moment of the regular hexagon.
\end{proposition}
\begin{proof}
The proof is provided in \cite[Appendix D]{massivemimotechrep}.
\end{proof}

Let us first consider the case of one dimension $d = 1$. According to Proposition \ref{interpreterprop}, the performance in (\ref{generalasymptoticformulad1}) is achieved by the uniform quantizer codebook 
\begin{align}
\label{oqpwiepqowie}
\left\{\frac{M}{2n}, \frac{3M}{2n},\ldots,\frac{(2n-1)M}{2n}\right\}, 
\end{align}
whose Voronoi cells intervals of length $M/n$. Since, according to (\ref{oiquweoiqwueoiquweoiquweq2}), maximizers of $\widetilde{R}(\cdot)$  approximately coincide with the maximizers of $R(\cdot)$, the codebook (\ref{oqpwiepqowie}) also corresponds to the approximately optimal set of antenna locations for the case of one dimension and uniform distribution of users. Substituting (\ref{generalasymptoticformulad1}) to (\ref{poqwjepqoweq}), we obtain
\begin{align}
\label{qp33woeupqweu}
R_n \simeq  (r-1)\log n + \log \frac{P}{M^r}  + r\zeta_1^{\star}.
\end{align}
For $d=1$ and a uniform distribution of users, we thus conclude that the approximately-optimal antenna configuration is (\ref{oqpwiepqowie}), and the corresponding best per-user rate is (\ref{qp33woeupqweu}).

For more than one dimension, the asymptotic behavior of $\widetilde{R}_n$ is provided by (\ref{generalasymptoticformula}). In particular, for two dimensions, the constant $\zeta_d^{\star}$ in (\ref{generalasymptoticformula}) can be determined to be the normalized logarithmic moment of the regular hexagon, as shown in (\ref{keowkeowke}). Correspondingly, a hexagonal arrangement of antenna groups is asymptotically optimal for maximizing $\widetilde{R}(\cdot)$. By (\ref{oiquweoiqwueoiquweoiquweq2}), we can then argue that a hexagonal arrangement is approximately optimal for maximizing the per-user rate.

For three or more dimensions, there is uncertainty regarding the value of $\zeta_d^{\star}$, and only the upper and lower bounds in (\ref{pwqoepoqwiewoqieqw}) are known. The optimal arrangement of antenna groups also remains to be determined. 
Still, for $d \geq 3$, the constant $\zeta_d^{\star} $ can be estimated by matching the performance of a numerically-designed maximizer of $\widetilde{R}(\cdot)$ to (\ref{generalasymptoticformula}) at some large $n$. 

\subsection{Solution to the Quantization Problem: Non-Uniform Case}
\label{secHighRes2}
We now work out the case of a general density $f$ using the idea of point density functions of high-resolution quantization theory \cite[Section IV]{quantization}. Specifically, given any set of antenna locations $\mathcal{X}$, consider the existence of a point density function $\lambda : \mathbb{R}^d \rightarrow \mathbb{R}_{\geq 0}$ such that the cube $[x,x+\mathrm{d}x]$ contains $n\lambda(x)\mathrm{d}x$ antenna locations and $\int_S \lambda(x)\mathrm{d}x = 1$. Since $f$ is approximately uniform on $[x,x+\mathrm{d}x]$, the average per-user rate conditioned on the users being in $[x,x+\mathrm{d}x]$ is then
\begin{align}
\frac{1}{d}\log (n\lambda(x)) +  \zeta_d^{\star} + o(1),
\end{align}
according to (\ref{generalasymptoticformula}). Averaging out the user density, we obtain
\begin{align}
\widetilde{R}(\mathcal{X}) = \frac{1}{d}\log n +   \zeta_d^{\star}+ \frac{1}{d}\int_S f(x) \log \lambda(x) \mathrm{d}x + o(1).
\end{align}
Maximizing $\widetilde{R}(\mathcal{X})$ over all $\mathcal{X}$ is equivalent to maximizing $\widetilde{R}(\mathcal{X})$ over all point density functions $\lambda(x)$ subject to the normalization $\int_S \lambda(x)\mathrm{d}x = 1$. Let
\begin{align}
H(f) \triangleq \int_S f(x) \log \frac{1}{f(x)} \mathrm{d}x
\end{align}
denote the differential entropy of the user density. According to Gibbs' inequality, we have
\begin{align}
\label{oqiwue123}
\widetilde{R}(\mathcal{X}) \leq \frac{1}{d}\log n +   \zeta_d^{\star} - \frac{1}{d} H(f)  + o(1).
\end{align}
Equality is achieved when $\lambda(x) = f(x),\,\forall x\in S$, i.e., when the antenna density matches the user density.

\subsection{A Closed-Form Formula for the Best Per-User Rate}
We can now combine (\ref{oqiwue123}) with (\ref{oiquweoiqwueoiquweoiquweq2}) to obtain the following main result of this section.
\begin{theorem}
Let $r > d$. For large $n$ and $r$, the best per-user rate can be approximated by
\begin{align}
\label{mainformula}
R_n \simeq \left(\frac{r}{d}-1\right)\log n + \log P+ r \zeta_d^{\star} - \frac{r}{d} H(f).
\end{align}
Such a performance is achievable when the density of antenna groups is the same as the user density. 
\end{theorem}

Even though the approximations to obtain (\ref{mainformula}) are shown to be valid for large $n$ and $r$, numerical results in the next section show that (\ref{mainformula}) is very accurate even for small $n$ and/or $r$. In particular, we can observe that the best per-user rate decreases linearly with the differential entropy of the user density. Since a uniform density maximizes the differential entropy, the worst per-user rate is achieved for a uniform distribution of users. Also, according to (\ref{mainformula}), the maximum achievable macromultiplexing gain of the massive MIMO system is $(\frac{r}{d}-1)^+$, as also claimed in Section \ref{secAchievability}. A formal derivation of (\ref{mainformula}) or its refinements are left as future work. 

\section{Numerical Results}
\label{secNumerical}
In this section, we present numerical simulations that verify our analytical results. We have numerically optimized the per-user rate $R(\mathcal{X})$ in (\ref{actualrakakakeks}) using gradient ascent, and compared the resulting per-user rate with the closed-form formula (\ref{mainformula}).

We present simulation results for the one-dimensional cell $S = [0,1]$ in Fig. \ref{simfig1}. The horizontal axis represents the number of distinct antenna locations $n$. The vertical axis represents the best-possible per-user rate $R_n$. We have considered the three cases $r\in\{2,4,8\}$ for the path loss exponent. Also, we have considered a uniform distribution $f_{1}(u) \triangleq \mathbf{1}(u\in[0,1])$, as well as a Beta distribution $f_{2}(u) \triangleq 30u(1-u)^4 f_1(u)$ of users with $H(f_2) = \frac{35}{12}  - \log 30 = -0.48453\cdots$. The label ``analysis'' refers to our general formula (\ref{mainformula}) for the best per-user rate, while the label ``simulation'' refers to the per-user rate we have obtained through the gradient ascent procedure.

\begin{figure}
\begin{center}
\scalebox{1.2}{\includegraphics{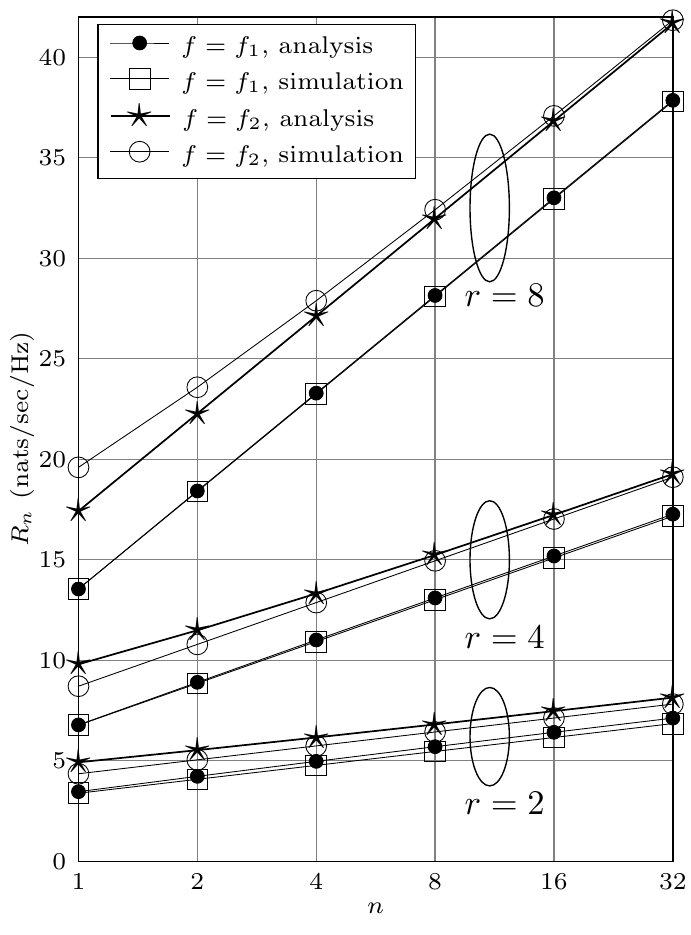}}
\end{center}
\vspace{-10pt}
\caption{Best per-user rates for a one-dimensional cell.}
\label{simfig1}
\end{figure}

Our analysis has suggested that (\ref{mainformula}) will be accurate for large $n$ and $r$. The results of Fig. \ref{simfig1} suggest that our formula (\ref{mainformula}) agrees well with the simulations even for small $n$ and $r$. For example, for $r=2$, the difference between analysis and simulation is at most $0.6$ nats/sec/Hz. For $r\in\{4,8\}$ and a uniform distribution, the analysis is almost an exact match to the simulation with barely noticeable differences. For the Beta distribution of users, the analysis converges to an exact match only as $n$ increases. In general, for non-uniform distributions, a large-enough $n$ is necessary for the accuracy of (\ref{mainformula}) so that the high-resolution arguments in Section \ref{secHighRes2} become valid. Also, in all the $6$ cases, the slope of analysis curves matches their respective simulation curves, indicating that the macromultiplexing gain of the system is indeed $r-1$.

We present simulations results for two-dimensional cells in Fig. \ref{simfig2}. We have considered a uniform distribution $f_3(u) \triangleq  \mathbf{1}(u\in[0,1]^2)$ of users over the square cell $S = [0,1]^2$. In order to demonstrate that (\ref{mainformula}) can also be applied to unbounded cells, we have also considered a Gaussian distribution $f_4(u) = \frac{1}{\pi}e^{-\|u\|^2}$ with $S = \mathbb{R}^2$ and $H(f_4) = \log(e\pi)$. We can observe that, as in the case of $d=1$ in Fig. \ref{simfig1}, the analysis is almost an exact match to the simulation for $r=8$ and a uniform distribution. For $r=8$ and the Gaussian distribution, the analysis correctly predicts the asymptotic behavior of the best per-user rate. In general, compared to the results in Fig. \ref{simfig1} for $d=1$, the gaps between the analysis and the corresponding simulation results are larger for $d=2$. The reason for the larger mismatch is that a point $u$ has potentially more neighboring antenna locations in a higher dimensional cell. As a result, for a given finite $r$,  the approximation (\ref{lattlemma2}) will be off by a larger multiplicative factor as $d$ increases. Still, the simulations agree well with the analysis even when $r$ is as low as $4$, and correctly predict the maximum achievable macromultiplexing gains at any $r$. In particular, according to (\ref{mainformula}), and as can also be observed in Fig. \ref{simfig2}, the maximum macromultiplexing gain for $r=d=2$ is $0$, meaning that the per-user rate will be bounded even as $n\rightarrow\infty$.

\begin{figure}
\begin{center}
\scalebox{1.2}{\includegraphics{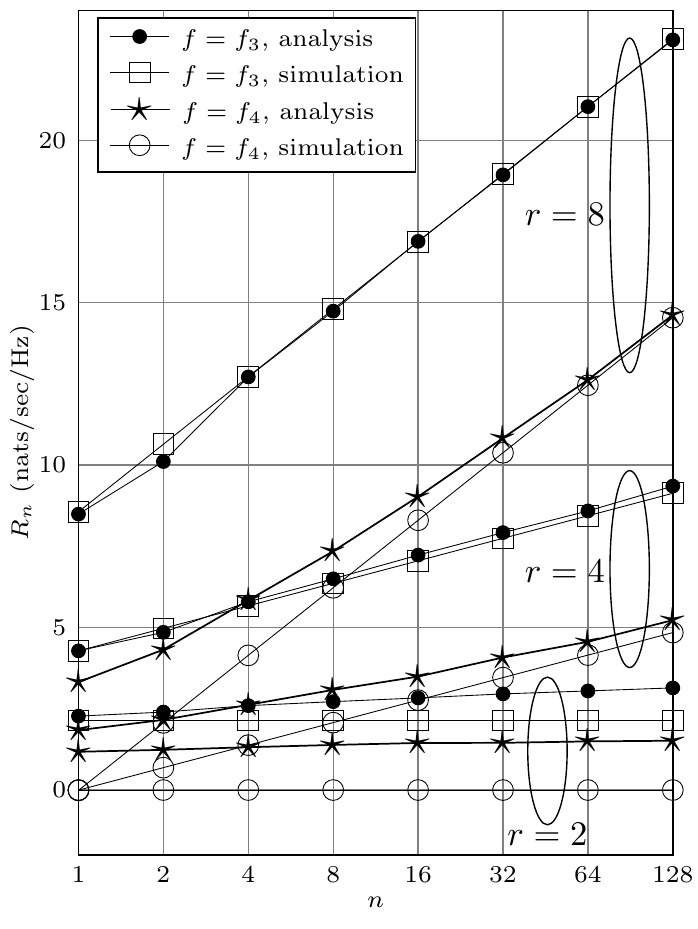}}
\end{center}
\vspace{-10pt}
\caption{Best per-user rates for a two-dimensional cell.}
\label{simfig2}
\end{figure}

\section{Conclusions}
\label{secConclusions}
We have studied a massive MIMO uplink where the base station antennas are evenly distributed to $n$ different locations of the cell. We have shown that a per-user rate of $O(\log n)$ is achievable through optimization of antenna locations. We have characterized the corresponding ``macromultiplexing gains'' in terms of the ambient dimension of the cell and the path loss exponent. Also, using tools from high-resolution quantization theory, we have found an accurate formula for the achievable per-user rate for any distribution of users throughout the cell.


\begin{thebibliography}{1}
\bibitem{marzetta1}
T. L. Marzetta, ``Noncooperative cellular wireless with unlimited numbers of base station antennas,'' \emph{IEEE Trans. Wireless Commun.}, vol. 9, no. 11, pp. 3590--3600, Nov. 2010.
\bibitem{swindletutorial}
L. Lu, G. Y. Li, A. L. Swindlehurst, A. Ashikhmin, and R. Zhang, ``An overview of massive MIMO: Benefits and challenges.'' \emph{IEEE J. Select. Areas Signal Process.}, vol. 8, no. 5, Oct. 2014.
\bibitem{saleh_dist_antennas}
A. A. M. Saleh, A. J. Rustako, Jr., and R. S. Roman, ``Distributed antennas for indoor radio communications,'' \emph{IEEE Trans. Commun.}, vol. COM-35, no. 12, pp. 1245--1251, Dec. 1987.
\bibitem{kjkerpez}
K. J. Kerpez, ``A radio access system with distributed antennas,'' \emph{IEEE Trans. Veh. Tech.}, vol. 45, no. 2, pp. 265--275, May. 1996.
\bibitem{wanchoi1}
W. Choi, J. G. Andrews, ``Downlink performance and capacity of distributed antenna systems in a multicell environment,'' \emph{IEEE Trans. Wireless Commun.}, vol. 6, no. 1, pp. 69--73, Jan. 2007.
\bibitem{j_gan_circular_das2}
J. Gan, S. Zhou, J. Wang, and K. Park, ``On the sum-rate capacity of multi-user distributed antenna system with circular antenna layout,'' \emph{IEICE Trans. Commun.}, vol. E89-B, no. 9, pp. 2612--2616, Sept. 2006.
\bibitem{j_gan_circular_das}
J. Gan, Y. Li, L. Xiao, S. Zhou, and J. Wang, ``On sum rate and power consumption of multi-user distributed antenna system with circular antenna layout,'' \emph{EURASIP J. Wireless Commun. Networking}, vol. 2007, Article ID 89780.
\bibitem{yindi_jing_circular_das}
A. Yang, Y. Jing, C. Xing, Z. Fei, and J. Kuang, ``Performance analysis and location optimization for massive MIMO systems with circularly distributed antennas,'' \emph{IEEE Trans. Wireless Commun.}, vol. 14, no. 10, pp. 5659--5671, Oct. 2015.
\bibitem{kamga1}
G. N. Kamga, M. Xia, and S. Aissa, ``Spectral-efficiency analysis of massive MIMO systems in centralized and distributed schemes,'' \emph{accepted to IEEE Trans. Commun.}, Jan. 2016.
\bibitem{wang_ant_loc_opt}
X. Wang, P. Zhu, and M. Chen, ``Antenna location design for generalized distributed antenna systems,'' \emph{IEEE Commun. Letters}, vol. 13, no. 5, pp. 315--317, May 2009.
\bibitem{wang_ant_loc_opt2}
Y. Qian, M. Chen, X. Wang, and P. Zhu, ``Antenna location design for distributed antenna systems with selective transmission,'' \emph{Intl. Conf. Wireless Commun. Signal Process.}, Nov. 2009.
\bibitem{firouz1}
S. Firouzabadi and A. Goldsmith, ``Optimal placement of distributed antennas in cellular systems,'' \emph{IEEE Intl. Workshop on Signal Process. Advances in Wireless Commun.}, June 2011.
\bibitem{circular_das_with_central_antenna}
E. Park, S.-R. Lee, and I. Lee, ``Antenna placement optimization for distributed antenna systems,'' \emph{IEEE Trans. Wireless Commun.}, vol. 11, no. 7,  pp. 2468--2477, July 2012.
\bibitem{lin_dai}
L. Dai, ``A comparative study on uplink sum capacity with co-located and distributed antennas,'' \emph{IEEE J. Select. Areas Commun.}, vol. 29, no. 6, pp. 1200--1213, June 2011. 
\bibitem{liu_antenna}
J. Liu, J. Wu, G. Wei, and W. Li, ``Antenna location optimization for hybrid user distribution,'' \emph{IEEE Intl. Conf. Commun. Workshop}, June 2015.
\bibitem{zhang} 
J. Zhang, C.-K. Wen, S. Jin, X. Gao, and K.-K. Wong, ``On capacity
of large-scale MIMO multiple access channels with distributed sets
of correlated antennas,'' \emph{IEEE J. Sel. Areas Commun.}, vol. 31, no. 2,
pp. 133--148, Feb. 2013.
\bibitem{massivemimotechrep}
E. Koyuncu, ``Performance gains of optimal antenna deployment for massive MIMO systems,'' \emph{Tech. Rep.}, Apr. 2017. [Online] Available: https://webfiles.uci.edu/ekoyuncu/alotr.pdf
\bibitem{quantization}
R. M. Gray and D. L. Neuhoff, ``Quantization," \emph{IEEE Trans. Inf. Theory}, vol. 44, no. 6, pp. 2325--2383, Oct. 1998.
\bibitem{zador1}
P. L. Zador, ``Asymptotic quantization error of continuous signals and the
quantization dimension,'' \emph{IEEE Trans. Inf. Theory}, vol. 28, no. 2, pp.
139--148, Mar. 1982.
\end{thebibliography}
\end{document}